\documentclass[12pt, draftcls, onecolumn]{IEEEtran}
\usepackage{subfigure,cite,graphicx,amsmath,amssymb,mathrsfs,epsf}\usepackage{multirow}
\usepackage{cite,graphicx,epsfig,subfigure,amsmath,amssymb,mathrsfs,epsf,graphics,color,slashbox}
\usepackage{psfrag}
\newtheorem{definition}{Definition}
\newtheorem{theorem}{Theorem}
\newtheorem{lemma}{Lemma}

\newtheorem{corollary}{Corollary}
\newtheorem{remark}{Remark}

\def \P{\operatorname{P}}

\begin{document}
\title{Capacity of a Class of Multicast Tree Networks}
\author{\authorblockN{Si-Hyeon Lee and Sae-Young Chung\\}
\authorblockA{Department of EE, KAIST, Daejeon, Korea\\
Email: sihyeon@kaist.ac.kr, sychung@ee.kaist.ac.kr} \thanks{The material in this paper was presented in
part at the Information Theory and Applications Workshop, UCSD, San Diego, CA, USA, January/February 2010, at the IEEE International Symposium on Information Theory, Austin, TX, USA, June 2010, and at the Allerton Conference on Communication, Control, and Computing, Monticello, IL, USA, Sep. 2010. }} \maketitle

\maketitle
\begin{abstract}
In this paper, we characterize the capacity of a new class of single-source multicast discrete memoryless relay networks having a tree topology in which the root node is the source and each parent node in the graph has at most one noisy child node and any number of noiseless child nodes. This class of multicast tree networks includes the class of diamond networks studied by Kang and Ulukus as a special case, where they showed that the capacity can be strictly lower than the cut-set bound.  For achievablity, a novel coding scheme is constructed where each noisy relay employs a combination of decode-and-forward (DF) and compress-and-forward (CF) and each noiseless relay performs a random binning such that codebook constructions and relay operations are independent for each node and do not depend on the network topology. For converse,  a new technique of iteratively manipulating inequalities exploiting the tree topology is used. 
\end{abstract}


\IEEEpeerreviewmaketitle


\section{Introduction}\label{sec:introduction}
In this paper, we consider a single-source multicast discrete memoryless relay network in which the source wants to send the same message reliably to multiple destinations with the help of one or more relays. A model of relay networks was introduced by van der Meulen in~\cite{meulen_thesis,Meulen:71}. However, the single-letter capacity characterization has been open even for three-node relay networks, i.e., relay networks having a source, a relay, and a destination. In their seminal paper~\cite{Gamal:79}, Cover and El Gamal developed two fundamental coding strategies for three-node relay networks. One of them is decode-and-forward (DF), where the relay decodes the message and forwards it to the destination, which was shown to be optimal for physically degraded channels \cite{Gamal:79}. DF was generalized for multiple relays in~\cite{XieKumar:05, KramerGastparGupta:05}. In another strategy, compress-and-forward (CF), the relay compresses its received block and sends the compressed information to the destination. CF was shown to achieve the capacity for some classes of relay networks~\cite{Kim:08,Yu:09}.
Recently, CF was generalized to noisy network coding in \cite{Lim:10} for multiple relays, which includes many previous results on relay networks \cite{Gamal:79,Yeung:00,Effros:06,AvestimehrDiggaviTse:11} as special cases. A potentially better strategy is to decode as much as possible and compress the residual information, i.e., a combination of DF and CF~\cite{Gamal:79}. Indeed such a strategy was shown to be optimal by Kang and Ulukus for a certain class of diamond networks in~\cite{KangUlukus:11}, which consists of a source, a noisy relay, a noiseless relay that receives exactly what the source sends, and a destination that has orthogonal finite-capacity links from relays. For this class of diamond networks, it was shown that a combination of DF and CF at the noisy relay is optimal and the cut-set bound is in general loose~\cite{KangUlukus:11}.

In this paper, we show the optimality of a combination of DF and CF for a new class of single-source multicast relay networks with an arbitrary number of nodes, which includes the class of diamond networks in \cite{KangUlukus:11} as a special case. In this class, which we call multicast tree networks, a network has a tree topology in which the root node is the source and each parent node in the graph has at most one noisy child node and any number of noiseless child nodes. We note that the achievability and converse for diamond networks in \cite{KangUlukus:11} cannot be directly generalized to those for our multicast tree networks. First, the codebook constructions and relay operations of the coding scheme in \cite{KangUlukus:11} for diamond networks, which has a single destination,  vary according to the link capacities from relays to the destination. This cannot be used for multicast tree networks since they have arbitrarily many destinations. Next, it would not be easy to generalize the converse proof technique in \cite{KangUlukus:11} for diamond networks, which have only four nodes in three levels, for our multicast tree networks, which have arbitrarily many nodes in arbitrarily high levels. Therefore, for these two reasons, we need new techniques.
The key technical contributions in the achievability and converse in this paper are as follows:
\begin{itemize}
\item Achievability: For the generalization to multicast tree networks, we construct a \emph{robust} coding scheme where codebook constructions and relay operations are independent for each node and do not depend on the network topology. Such a robustness of the coding scheme makes the generalization from a single destination to multiple destinations possible.

\item Converse: To get a very simple min-cut expression, we use a novel technique of iteratively manipulating inequalities, i.e., we recursively reduce a number of inequalities into one using the tree topology.
\end{itemize}

The organization of this paper is as follows. The model of a class of multicast tree networks is presented in Section \ref{sec:model}. In Section \ref{sec:main}, we present lower and upper bounds on the capacity of the class of multicast tree networks and show a condition for these two bounds to coincide. In Section \ref{sec:achievability}, we derive the lower bound by presenting a coding scheme where each noisy relay employs a combination of DF and CF and each noiseless relay performs a random binning. In Section \ref{sec:ub}, the upper bound is shown using a recursion exploiting the tree topology. In Section \ref{sec:diamond}, we present an equivalent capacity expression for diamond networks that shows that without loss of optimality we can construct the coding scheme such that what is compressed after decoding at a noisy relay is a noisy observation of almost uncoded information. The conclusion of this paper is given in Section~\ref{sec:conclusion}.

The following notations will be used in the paper.  For two integers $i$ and $j$, $[i:j]$ denotes the set $\{i,i+1, \ldots, j\}$, $x_{i}^j$ denotes a row vector $(x_{i},x_{i+1}, ...., x_{j})$, and $x^j$ denotes $x_1^j$. $x_{S}$ for a set $S$ denotes a row vector $(x_i:i\in S)$. According to the context, $k$ sometimes denotes the single-element set $\{k\}$ for notational convenience.

In this paper, we follow the notion of $\epsilon$-robustly typical sequence introduced in~\cite{Orlitsky:01}. Let $N_{x^n}(x)$ denote the number of occurrences of $x\in \mathcal{X}$ in the sequence $x^n$. Then, $x^n$ is said to be $\epsilon$-robustly typical (or just typical) for $\epsilon>0$ if for every $x\in \mathcal{X}$,
\begin{align*}
\bigg|\frac{N_{x^n}(x)}{n}-p(x)\bigg|\leq \epsilon p(x) .
\end{align*}
The set of all $\epsilon$-robustly typical $x^n$ is denoted as $T_{\epsilon}(X)$, which is shortly denoted as $T_{\epsilon}$. Similarly, let $N_{x^n, y^n}(x,y)$ denote the number of occurrences of $(x,y)\in \mathcal{X}\times \mathcal{Y}$ in the sequence $(x^n, y^n)$. The sequence $(x^n,y^n)$ is said to be $\epsilon$-robustly typical (or just typical) if
\begin{align*}
\bigg|\frac{N_{x^n,y^n}(x,y)}{n}-p(x,y)\bigg|\leq \epsilon p(x,y)
\end{align*}
for every $(x,y)\in\mathcal{X}\times \mathcal{Y}$. The set of all $\epsilon$-robustly typical $(x^n, y^n)$ is denoted by $T_{\epsilon}(X,Y)$ or $T_{\epsilon}$ in short.

\section{Model}\label{sec:model}
A single-source multicast discrete memoryless relay network of $N$ nodes
\begin{align*}
\left(\mathcal{X}_1\times ... \times \mathcal{X}_N, p(y_1, ..., y_N|x_1, ..., x_N), \mathcal{Y}_1\times ... \times \mathcal{Y}_N \right)
\end{align*}
consists of alphabets $\mathcal{X}_k, \mathcal{Y}_k$ for $k\in[1:N]$ and a collection of conditional probability mass functions $p\left(y_1,...,y_N|x_1,...,x_N\right)$ where $x_k\in \mathcal{X}_k$ and $y_k\in\mathcal{Y}_k$ for $k\in[1:N]$. Let $K$ denote the number of destinations. Let $1$ and $D_d$ denote the source and the set of nodes that forms the $d$-th destination, respectively, and let $\mathcal{Y}_1=\mathcal{X}_{D_d}=\emptyset$ for $d\in [1:K]$. We note that $D_d$ for $d\in [1:K]$ are not necessarily disjoint. Let $D\triangleq \bigcup_{d\in [1:K]}D_d$.

A $\left(2^{nR},n\right)$ code for a single-source multicast discrete memoryless relay network of $N$ nodes consists of a message set $\mathcal{W}_1=[1:2^{nR}]$, a source encoder that assigns a codeword $x_1^n(w_1)$ to each message $w_1\in \mathcal{W}_1$, a set of relay encoders, where encoder $k\in [2:N]\setminus D$ assigns a symbol $x_{k,i}(y_k^{i-1})$ to every received sequence $y_k^{i-1}$ for $i\in [1:n]$, and a set of decoders, where decoder $k\in[1:K]$ assigns an estimate $\hat{w}_{1,k}$ to each received sequence $y_{D_k}^n$. The message $W_1$ is chosen uniformly from the set $\mathcal{W}_1$.
The average probability of error for a $(2^{nR}, n)$ code is given as
\begin{align*}
P_e^{(n)}&\triangleq \P\left\{\hat{W}_{1,d}\neq W_1 \mbox{ for some } d\in [1:K] \right\}.
\end{align*}
A rate $R$ is said to be \emph{achievable} if there exists a sequence of $(2^{nR}, n)$ codes such that $P_e^{(n)}\rightarrow 0$ as $n\rightarrow \infty$. The capacity is the supremum of all achievable rates.

A single-source multicast discrete memoryless relay network is called a multicast tree network if the probability distribution has the form of
\begin{align*}
p\left(y_1,...,y_N|x_1,...,x_N\right)=\prod_{k\in [1:N]}p\left(y_k|x_{p_k}\right) 
\end{align*}
where $p_k$ is called the \emph{parent node} of node $k$ and $k$ is called a \emph{child node} of node $p_k$. A child node is considered to be one level lower than its parent node. A node without a parent node is called the \emph{root node} and a node that has no child node is called a \emph{leaf node}. Let $L_k$ for $k\in [1:N]$ denote the set of leaf nodes that branches out from node $k$. For tree $T$, let $T_k$ for $k\in [1:N]$ denote the \emph{subtree} of $T$ that consists of node $k$ and all of its descendants in $T$. 

In this paper, our goal is to present lower and upper bounds on the capacity of a class of multicast tree networks and to find some tightness conditions of those two bounds. In this class of multicast tree networks,  the source node is the root node, $D_d\subseteq L_1$ for $d\in[1:K]$, and each parent node has at most one noisy child node and any number of noiseless child nodes, i.e., $y_k=x_{p_k}$ if $k$ is a noiseless child node of node $p_k$. Without loss of generality, we assume that $D=L_1$. Let $G_d\triangleq \{k|L_k\cap D_d\neq \emptyset\}$ for $d\in[1:K]$. Let $n_k$ and $M_k$ for $k\in [1:N]$ denote the noisy child node and the set of noiseless child nodes of node $k$, respectively. Let $Z_k$ for $k\in [1:N]$ denote the set of child nodes of node $k$, i.e., $Z_k=n_k\cup M_k$. From now on, we only consider this class of multicast tree networks. See Fig.~\ref{fig:tree_ex}.
\begin{figure}[t]
 \centering
  {\small
\psfrag{1}[c]{$1$}
\psfrag{2}[c]{$2$}
\psfrag{3}[c]{$3$}
\psfrag{4}[c]{$4$}
\psfrag{5}[c]{$5$}
\psfrag{6}[c]{$6$}
\psfrag{7}[c]{$7$}
\psfrag{8}[c]{$8$}
\psfrag{9}[c]{$9$}
\psfrag{10}[c]{$10$}
\psfrag{11}[c]{$11$}
\psfrag{12}[c]{$12$}
\psfrag{13}[c]{$13$}
\psfrag{14}[c]{$14$}

  \includegraphics[width=100mm]{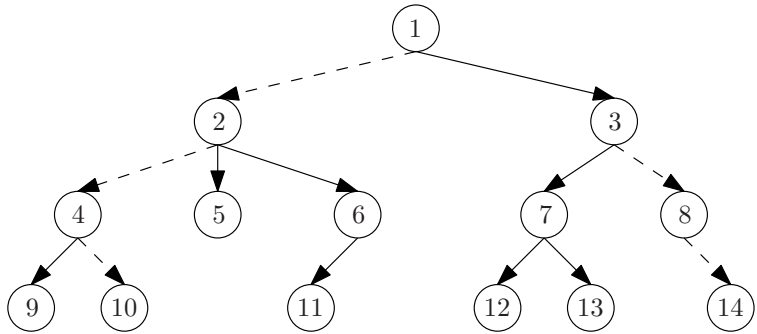}}
  \caption{An example of our multicast tree networks. The solid and dashed lines represent noiseless and noisy links, respectively.  In this example, the parent node of node 3 is node 1 and the child nodes of node 3 are nodes 7 and 8. Node 1 is the root node and nodes 5, 9, 10, 11, 12, 13, and 14 are the leaf nodes.   A destination is a subset of leaf nodes. For instance, destination 1 is the set of nodes 5, 11, 12, and 13, destination 2 is the set of nodes 9, 12, and 14, and destination 3 is node 10. $L_2$ is the set of nodes 5, 9, 10, and 11. $T_3$ is the subtree that consists of nodes 3, 7, 8, 12, 13, and 14.} \label{fig:tree_ex}
\end{figure}

A practical example of our  multicast tree networks is depicted in Fig. \ref{fig:sensor}, which represents a sensor network where a sensor node wants to send a message to the gateway nodes at the boundary connected with infinite-capacity wired links. In this example, each relay node has outgoing links to its neighbor relays such that one of the links is arbitrarily noisy and the others are noiseless. Motivation for assuming noiseless links comes from a practical scenario where a transmitter is using a fixed modulation scheme tuned for the worst link and thus the transmission from the transmitter to the other receivers with better channel qualities looks almost noiseless.
\begin{figure}[t]
 \centering
  {\includegraphics[width=70mm]{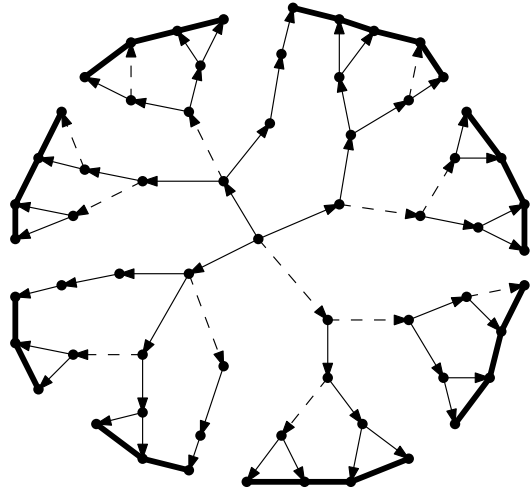}}
  \caption{A sensor network in which a sensor node wants to send a message to gateway nodes at the boundary connected with infinite-capacity wired links. The solid and dashed lines represent noiseless and noisy links, respectively, and thick lines at the boundary represent infinite-capacity wired links. } \label{fig:sensor}
\end{figure}

\section{Main Results for Multicast Tree Networks} \label{sec:main}
Let us present lower and upper bounds on the capacity of multicast tree networks.
\begin{theorem} \label{thm:tree_multicast}
The capacity $C$ of multicast tree networks is lower- and upper-bounded as
\begin{align}
C&\geq \max_{\prod_{k\in [1:N]}p(u_k,x_k)p(\hat{y}_{n_k}|u_k,y_{n_k})} \min_{d\in[1:K]}\min_{S_d} \sum_{k\in A_{S_d,d}}I(U_{k};Y_{n_k})+H(X_{k}|U_{k})\cr
&~~~~~~~~~~~~~+\sum_{k\in B_{S_d,d}}I(U_{k};Y_{n_k})+I(X_{k};\hat{Y}_{n_k}|U_{k})-\sum_{k\in C_{S_d,d}}I(Y_{n_k};\hat{Y}_{n_k}|U_{k}, X_{k}) \label{eqn:multicast_ach}\\
C&\leq \max_{\prod_{k\in [1:N]}p(u_k,x_k)} \min_{d\in[1:K]}\max_{\prod_{k \in [1:N]}p(\hat{y}_{n_k}|u_k,y_{n_k})}\min_{S_d} \sum_{k\in A_{S_d,d}}I(U_{k};Y_{n_k})+H(X_{k}|U_{k})\cr
&~~~~~~~~~~~~~+\sum_{k\in B_{S_d,d}}I(U_{k};Y_{n_k})+I(X_{k};\hat{Y}_{n_k}|U_{k})-\sum_{k\in C_{S_d,d}}I(Y_{n_k};\hat{Y}_{n_k}|U_{k}, X_{k})\label{eqn:multicast_cov}
\end{align}
over all cuts $S_d\subset G_d$ such that $1\in S_d$, $D_d\subseteq S_d^c$, $M_k\cap G_d\subset S_d$ if $n_k\in S_d$, and $p_k\in S_d$ if $k\in S_d$ with cardinalities of alphabets such that
\begin{subequations}\label{eqn:cardinality}
\begin{align}
|\mathcal{U}_k|&\leq |\mathcal{X}_k|+4\\
|\mathcal{\hat{Y}}_{n_k}|&\leq |\mathcal{U}_k||\mathcal{Y}_{n_k}|+2\leq |\mathcal{X}_k||\mathcal{Y}_{n_k}|+4|\mathcal{Y}_{n_k}|+2
\end{align}
\end{subequations}
for $k\in [1:N]$.
Here, $A_{S_d,d}, B_{S_d,d},$ and $C_{S_d,d}$ for $d\in [1:K]$ denote the following disjoint subsets of $S_d$.
\begin{align*}
A_{S_d,d}&\triangleq\{k|k\in S_d, Z_k \subseteq S_d^c, M_k\cap G_d\neq \emptyset\}\\
B_{S_d,d}&\triangleq\{k|k\in S_d, n_k\in S_d^c, M_k\cap G_d\subset S_d,  n_k\cap G_d\neq \emptyset \}\\
C_{S_d,d}&\triangleq\{k|k\in S_d, Z_k\cap G_d\subset S_d\}
\end{align*}
See Table~\ref{table:cases}.
\begin{table*}
\caption{Classification of $k\in {S_d}$ into $A_{S_d,d}$, $B_{S_d,d}$, and $C_{S_d,d}$}\begin{center}
  \begin{tabular}{|c c|c|c|c|}
    \hline
    \multicolumn{2}{|c|}{\multirow{2}{*}{\backslashbox{$M_k\cap G_d$}{$n_k\cap G_d$}}} &\multicolumn{2}{|c|}{$n_k\cap G_d\neq \emptyset$} &\multirow{2}{*}{$n_k\cap G_d=\emptyset$}\\
    \cline{3-4}
    &&$n_k\cap G_d\subset S_d$ & $n_k\cap G_d\subseteq S_d^c$& \\
    \hline
    \multicolumn{1}{|c|}{\multirow{2}{*}{$M_k\cap G_d\neq \emptyset$}} & $M_k\cap G_d\subset S_d$ & $k\in C_{S_d,d}$ & $k\in B_{S_d,d}$ & $k\in C_{S_d,d}$\\
    \cline{2-5}
    \multicolumn{1}{|c|}{} &$M_k\cap G_d\subseteq S_d^c$ & -- & $k\in A_{S_d,d}$ & $k\in A_{S_d,d}$\\
    \hline
    \multicolumn{2}{|c|}{$M_k\cap G_d= \emptyset$}&$k\in C_{S_d,d}$& $k\in B_{S_d,d}$& -- \\
    \hline
      \end{tabular}\label{table:cases}
\end{center}
``--" indicates that corresponding cases do not happen for a cut $S_d$ of interest.
\end{table*}

\end{theorem}

\begin{remark}
In Theorem~\ref{thm:tree_multicast}, a cut $S_d$ of interest for destination $d\in [1:K]$ satisfies that $p_k\in S_d$ if $k\in S_d$ and $M_k\cap G_d\subset S_d$ if $n_k\in S_d$ in addition to that $1\in S_d$ and $D_d\subseteq S_d^c$. This additional condition signifies that node $p_k$ can decode whatever node $k$ can and a node in $M_k$ can decode whatever node $n_k$ can.
\end{remark}

We can see that the lower and upper bounds in Theorem~\ref{thm:tree_multicast} meet when the maximizing distribution of $\prod_{k \in [1:N]}p(\hat{y}_{n_k}|u_k,y_{n_k})$ is independent of destinations. The following corollary presents a class of such multicast tree networks. Let $a_d$ for $d\in [1:K]$ denote the node at the lowest level in the set $\{k|D_d\subseteq L_k\}$. The proof is in Appendix \ref{appendix:tightness}.
\begin{corollary}\label{corollary:tightness}
If $L_{a_i}\cap D_j=\emptyset$ for all $i,j\in [1:K]$ such that $i\neq j$, the lower and upper bounds in Theorem \ref{thm:tree_multicast}  coincide.
\end{corollary}
Corollary \ref{corollary:tightness} says that the lower and upper bounds meet when each set of nodes forming a destination is included in a disjoint subtree. For example, the lower and upper bounds for the multicast tree network represented in Fig.~\ref{fig:tree_ex} meet when destination 1 is the set of nodes 5, 9, 10, and 11, destination 2 is the set of nodes 12 and 13, destination 3 is node 14.

For the single destination case, the lower and upper bounds in Theorem \ref{thm:tree_multicast} coincide trivially. In this case, the following corollary gives a simpler capacity expression.
\begin{corollary} \label{corollary:cutset}
For tree networks with a single destination, the capacity is given as
\begin{align}
\max \min_S I( U_S;Y_{S^c} \setminus   X_S ) + I(X_S;\hat{Y}_{S^c}|U_S) - I(Y_S;\hat{Y}_S|U_S, X_S) \label{eqn:cap_theorem}
\end{align}
where the minimization is over all cuts $S\subset [1:N]$ such that $1\in S$, $D\subseteq S^c$, $M_k\subset S$ if $n_k\in S$, and $p_k\in S$ if $k\in S$, and the maximization is over the joint distribution of
\begin{align}
\prod_{k\in [1:N]}p(u_k,x_k)p(\hat{y}_{n_k}|u_k,y_{n_k})\label{eqn:joint_distr_general}
\end{align}
with cardinalities of alphabets satisfying (\ref{eqn:cardinality}) for $k\in [1:N]$. In (\ref{eqn:cap_theorem}), $\hat{Y}_j=X_k$ for $k\in [1:N]$ and $j\in M_k$ and $Y_{S^c} \setminus   X_S$ denotes the set
\begin{align*}
\{Y_j|j\in S^c, j\notin M_k \mbox{ for all } k\in S\}.
\end{align*}
\end{corollary}
\begin{proof}
For a cut $S$ of interest, we have
\begin{align*}
I(U_S;Y_{S^c}\backslash X_S) &=  \sum_{k\in A_{S,1}\cup B_{S,1}} I(U_k;Y_{n_k}) \cr
I(X_S;\hat{Y}_{S^c}|U_S) &=   \sum_{k\in A_{S,1}} I(X_k;X_k,\hat{Y}_{n_k}|U_k) +    \sum_{k\in B_{S,1}} I(X_k;\hat{Y}_{n_k}|U_k)\cr
&=\sum_{k\in A_{S,1}} H(X_k|U_k) + \sum_{k\in B_{S,1}} I(X_k;\hat{Y}_{n_k}|U_k)\cr
I(Y_S;\hat{Y}_S|U_S,X_S)&=  \sum_{k\in C_{S,1}}I(Y_{n_k};\hat{Y}_{n_k}|U_k, X_k)
\end{align*}
from the joint distribution (\ref{eqn:joint_distr_general}), which concludes the proof.
\end{proof}
Here $U$ corresponds to the part of a message intended to be decoded by a noisy relay and $\hat{Y}$ corresponds the compressed version of a received block.

In contrast, only CF is performed at relays in noisy network coding~\cite{Lim:10}, whose achievable rate for general single-source single-destination discrete memoryless relay networks is given as
\begin{align}
\max\min_S  I(X_S; \hat{Y}_{S^c} , Y_D|X_{S^c},Q) - I(Y_S;\hat{Y}_S|X^N,  \hat{Y}_{S^c} , Y_D,Q)\label{eqn:noisy_network}
\end{align}
where the minimization is over all cuts $S\subset [1:N]$ such that $1\in S$ and $D\subseteq S^c$ and the maximization is over the joint distribution of
\begin{align*}
p(q)\prod_{k\in [1:N]} p(x_k|q)p(\hat{y}_k|x_k, y_k, q).
\end{align*}
Note that (\ref{eqn:cap_theorem}) and (\ref{eqn:noisy_network}) are somewhat similar especially the parts involving $\hat{Y}$'s but (\ref{eqn:cap_theorem}) includes $U$'s due to DF.

\section{Achievability}\label{sec:achievability}
Fix a joint distribution of (\ref{eqn:joint_distr_general}).  Fix $\epsilon''>\epsilon'>0$ and fix $r_{k, a}\geq 0, r_{k,b}\geq 0$, and $r_{n_k,v}\geq 0$ for $k\in [1:N]\setminus D$.
\subsubsection{Codebook generation} For $k\in [2:N]$, the index set $\mathcal{W}_k$ of node $k$ is defined as
\begin{align*}
\mathcal{W}_k \triangleq
\begin{cases}
[1:2^{nr_{p_k,a}}]\times [1:2^{nr_{k,v}}] &\mbox{for $k=n_{p_k}$}\\
[1:2^{nr_{p_k,a}}]\times [1:2^{nr_{p_k,b}}]&\mbox{for $k\in M_{p_k}$}
\end{cases}.
\end{align*}

For $k\in [1:N]\setminus D$, generate the codebooks following the steps below.
\begin{itemize}
\item Consider a random mapping $\gamma_k$ from $\mathcal{W}_k$ to $[1:2^{nr_{k,a}}]\times [1:2^{nr_{k,b}}]$ such that each $w_k\in \mathcal{W}_k$ is mapped to $\gamma_k(w_k)=(\alpha_k(w_k), \beta_k(w_k))$, where $\alpha_k(w_k)$ and $\beta_k(w_k)$ are uniformly and independently chosen from $[1:2^{nr_{k,a}}]$ and $[1:2^{nr_{k,b}}]$, respectively.

\item Generate $2^{nr_{k,a}}$ independent codewords $u_k^n(\alpha_k)$ for $\alpha_k\in [1:2^{nr_{k,a}}]$, of length $n$, according to $\prod_{i=1}^np(u_{k,i})$.

\item For each $\alpha_k\in [1:2^{nr_{k,a}}]$, generate $2^{nr_{k,b}}$ conditionally independent codewords $x_k^n(\beta_k|\alpha_k)$  for $\beta_k \in [1:2^{nr_{k,b}}]$, of length $n$, according to $\prod_{i=1}^np(x_{k,i}|u_{k,i}(\alpha_k))$.

\item For each $\alpha_k\in [1:2^{nr_{k,a}}]$, generate $2^{nr_{n_k,v}}$ conditionally independent codewords $\hat{y}_{n_k}^n(v_{n_k}|\alpha_k)$ for $v_{n_k}\in [1:2^{nr_{n_k,v}}]$, of length $n$, according to $\prod_{i=1}^np(\hat{y}_{n_k,i}|u_{k,i}(\alpha_k))$.

\item Let $x_k^n(w_k)$ denote $x_k^n(\beta_k|\alpha_k)$, where $(\alpha_k, \beta_k)=\gamma_k(w_k)$ for $w_k\in \mathcal{W}_k$.
\end{itemize}

The codebooks are revealed to all parties.
\subsubsection{Encoding at the source} For a message $w_1\in \mathcal{W}_1$, the source sends $x_1^n(w_1)$.

\subsubsection{Processing at node $k\in [2:N]$ such that $k=n_{p_k}$} Node $k$ operates following the steps below.
\begin{itemize}
\item Find a unique $\tilde{\alpha}_{p_k}$ such that
\begin{align*}
(u_{p_k}^n(\tilde{\alpha}_{p_k}), y_{k}^n)\in T_{\epsilon'}.
\end{align*}
If there is no such $\tilde{\alpha}_{p_k}$, randomly pick $\tilde{\alpha}_{p_k}\in[1:2^{nr_{p_k,a}}]$.

\item Seek for a $\tilde{v}_{k}$ such that
\begin{align*}
(u_{p_k}^n(\tilde{\alpha}_{p_k}),y_{k}^n,\hat{y}_{k}^n(\tilde{v}_{k}|\tilde{\alpha}_{p_k})) \in T_{\epsilon'}.
\end{align*}
If there are more than one such indices, randomly choose one among them. If there is no such $\tilde{v}_{k}$, randomly pick $\tilde{v}_{k}\in [1:2^{nr_{k,v}}]$.

\item Let $\tilde{w}_k=(\tilde{\alpha}_{p_k}, \tilde{v}_{k})$.

\item If $Z_k\neq \emptyset$, node $k$ sends $x_k^n(\tilde{w}_k)$.
\end{itemize}

\subsubsection{Processing at node $k\in [2:N]$ such that $k\in M_{p_k}$} Node $k$ operates following the steps below.
\begin{itemize}
\item Find a unique $(\tilde{\alpha}_{p_k}, \tilde{\beta}_{p_k})$ such that
\begin{align*}
x_{p_k}^n(\tilde{\beta}_{p_k}|\tilde{\alpha}_{p_k})=y_{k}^n.
\end{align*}
If there is no such $(\tilde{\alpha}_{p_k}, \tilde{\beta}_{p_k})$, randomly pick $(\tilde{\alpha}_{p_k}, \tilde{\beta}_{p_k})\in [1:2^{nr_{{p_k},a}}]\times [1:2^{nr_{{p_k},b}}]$.

\item Let $\tilde{w}_k=(\tilde{\alpha}_{p_k}, \tilde{\beta}_{p_k})$.

\item If $Z_k\neq \emptyset$, node $k$ sends $x_k^n(\tilde{w_k})$.
\end{itemize}

\subsubsection{Decoding at the destinations} 

The $d$-th destination for $d\in [1:K]$ decodes the message following the steps below.
\begin{itemize}
\item Construct a subset $F_{k,d}$ of $\mathcal{W}_k$ for every $k\in [1:N]$ in the following way. For $k\in D_d$, let $F_{k,d}\triangleq \{\tilde{w}_k\}$. For $k\notin G_d$, let $F_{k,d}\triangleq \mathcal{W}_k$. For all the other $k$'s, i.e., $k\in G_d\setminus D_d$, $F_{k,d}$'s are constructed recursively as
    \begin{align*}
    &F_{k,d}=\{w_k \big|(u_k^n(\alpha_k(w_k)),x_k^n(\beta_k(w_k)|\alpha_k(w_k)),\hat{y}_{n_k}^n({v_{n_k}}|\alpha_k(w_k)))\in T_{\epsilon''}, \cr
&(\alpha_k(w_k),v_{n_k})\in F_{n_k,d}, (\alpha_k(w_k),\beta_k(w_k))\in F_{j,d} \mbox{ for all $j\in M_k$ for some $v_{n_k}\in [1:2^{nr_{n_k,v}}]$}\}.
    \end{align*}

\item Find a unique $\hat{w}_{1,d}\in F_{1,d}$. If there is no such $\hat{w}_{1,d}$, randomly pick $\hat{w}_{1,d}\in \mathcal{W}_1$. The destination declares that $\hat{w}_{1,d}$ was sent.
\end{itemize}

\subsubsection{Analysis of the probability of error}
We analyze the probability of error for message $W_1$ averaged over the codebook ensemble. Let $\tilde{W}_k$ denote the chosen index at node $k$ for $k\in [2:N]$ and let $\tilde{V}_{n_k}$ denote the chosen covering index at node $n_k$ for $k\in [1:N]\setminus D$. Let us first introduce the notion of a supporting rate.
\begin{definition}
For our coding scheme, $T_k$ for $k\in [1:N]$ is said to \emph{support} a rate $r_k$ or have a supporting rate $r_k$ for destination $d\in [1:K]$ if, for any $\epsilon>0$,
\begin{align*}
\mu_{k,d}^{(n)}&\triangleq\P(\tilde{W}_k\notin F_{k,d})<\epsilon \cr
\nu_{k,d}^{(n)}&\triangleq \P(\tilde{w}_k'\in F_{k,d})<2^{-n(r_k-\epsilon)}
\end{align*}
for $\tilde{w}_k'\neq \tilde{W}_k$ for sufficiently small $\epsilon'$ and $\epsilon''$ and sufficiently large $n$.\footnote{$\P(\tilde{w}_k'\in F_{k,d})$ for all $\tilde{w}_k'\neq \tilde{W}_k$ are the same due to the symmetry of the codebook generation. } Note that the supremum of the supporting rate of $T_k$ for destination $d\in[1:K]$ becomes infinity and zero when $k\in D_d$ and $k\notin G_d$, respectively.
\end{definition}


The following lemma shows that $R<r_1$ is achievable if $T=T_1$ supports a rate  $r_1$ for all destinations.
\begin{lemma} \label{lemma:support_achievable}
If $T=T_1$ supports a rate  $r_1$ for all destinations, $R<r_1$ is achievable.
\end{lemma}
\begin{proof}
Fix $\epsilon>0$. If $T$ supports a rate  $r_1$ for all destinations, the average probability of error using our coding scheme is upper-bounded as
\begin{align}
P_e^{(n)}&=\P\left\{\hat{W}_{1,d}\neq W_1 \mbox{ for some } d\in [1:K] \right\}\cr
&\leq \sum_{d\in[1:K]}\P\left\{\hat{W}_{1,d}\neq W_1 \right\}\cr
&<\sum_{d\in [1:K]}\left(\mu_{1,d}^{(n)}+2^{nR}\nu_{1,d}^{(n)}\right)\cr
&<K\left(\epsilon+2^{-n(r_1-\epsilon-R)}\right) \label{eqn:support}
\end{align}
for sufficiently large $n$.
Note that (\ref{eqn:support}) is upper-bounded by $(K+1)\epsilon$ for sufficiently large $n$ if $R<r_1-\epsilon$. Thus, $R<r_1$ is achievable.
\end{proof}

Now, let us derive a sufficient condition for a supporting rate $r_1$ of $T$ for all destinations using the following lemma. The proof is at the end of this section.
\begin{lemma}\label{lemma:support}
Consider $d\in[1:K]$ and $k\in G_d \setminus D_d$. If $T_j$ for $j\in Z_k$ supports a rate $r_j$ for destination $d$, $T_k$ supports a rate $r_k$ for destination $d$ such that
\begin{subequations}\label{eqn:lemma_ach}
\begin{align}
r_k&\leq I(U_k;Y_{n_k})+H(X_k|U_k)\label{eqn:lemma_ach1}\\
r_k&\leq \sum_{j\in M_k\cap G_d}r_j+I(U_k;Y_{n_k})+I(X_k;\hat{Y}_{n_k}|U_k)\label{eqn:lemma_ach2}\\
r_k&\leq \sum_{j\in Z_k\cap G_d}r_j-I(Y_{n_k};\hat{Y}_{n_k}|U_k, X_k).\label{eqn:lemma_ach3}
\end{align}
\end{subequations}
\end{lemma}

To get a bound on the supporting rate $r_1$ of $T$ for destination $d\in [1:K]$ using Lemma~\ref{lemma:support}, we apply the Fourier-Motzkin elimination to the set of inequalities~(\ref{eqn:lemma_ach}) for all $k\in G_d\setminus D_d$  by removing all the other $r_k$'s, i.e., $k\in G_d\setminus D_d \setminus \{1\}$.\footnote{Note that $r_k$ for $k\in D_d$ is given by infinity. } The resultant inequalities of $r_1$ can be written as the min-cut form
\begin{align*}
r_1\leq  \min_{S_d} \sum_{k\in A_{S_d,d}}I(U_{k};Y_{n_k})+H(X_{k}|U_{k})+\sum_{k\in B_{S_d,d}}I(U_{k};Y_{n_k})+I(X_{k};\hat{Y}_{n_k}|U_{k})\\
-\sum_{k\in C_{S_d,d}}I(Y_{n_k};\hat{Y}_{n_k}|U_{k}, X_{k})
\end{align*}
where the minimization is over all cuts $S_d$ considered in Theorem~\ref{thm:tree_multicast}. Here, each cut $S_d$ corresponds to the set of inequalities that results in an inequality of $r_1$ in the Fourier-Motzkin elimination, i.e., the set of inequalities consists of (\ref{eqn:lemma_ach1}) for $k\in A_{S_d,d}$, (\ref{eqn:lemma_ach2}) for $k\in B_{S_d,d}$, and (\ref{eqn:lemma_ach3}) for $k\in C_{S_d,d}$.

For all destinations, we obtain the following sufficient condition for a supporting rate $r_1$.
\begin{align}
r_1\leq  \min_{d\in [1:K]}\min_{S_d} \sum_{k\in A_{S_d,d}}I(U_{k};Y_{n_k})+H(X_{k}|U_{k})+\sum_{k\in B_{S_d,d}}I(U_{k};Y_{n_k})+I(X_{k};\hat{Y}_{n_k}|U_{k})\cr
-\sum_{k\in C_{S_d,d}}I(Y_{n_k};\hat{Y}_{n_k}|U_{k}, X_{k}) \label{eqn:achievable_d}
\end{align}

From Lemma \ref{lemma:support_achievable}, all rates less than the right-hand side of~(\ref{eqn:achievable_d}) are achievable. By considering all joint distributions of~(\ref{eqn:joint_distr_general}), the lower bound in Theorem~\ref{thm:tree_multicast} is proved.

\subsubsection*{Proof of Lemma~\ref{lemma:support}} 
Fix $d\in[1:K]$ and $k\in G_d\setminus D_d$. Fix any $\epsilon>0$. Without loss of generality, assume that $\tilde{W}_k=(1,1)$ and $\gamma_k(1,1)=(1,1)$.
First, $\mu_{k,d}^{(n)}$ is upper-bounded as
\begin{align}
\mu_{k,d}^{(n)}&\leq \P\left(E_1 \cup E_2 \cup\bigcup_{j\in M_k} E_{3j}\right)\cr
&\leq \P\left(\tilde{E}_1\cup \tilde{E}_2 \cup \tilde{E}_3 \cup \tilde{E}_4\cup E_1\cup E_2 \cup \bigcup_{j\in M_k} E_{3j}\right)\cr
&\leq  \P(\tilde{E}_1)+\P(\tilde{E}_2)+\P(\tilde{E}_3)+\P(\tilde{E}_4)\cr
&~~~~~~~~~~~~~~~~~~~~~~~~~~+\P(E_1\cap\tilde{E}_1^c)+\P(E_2|\tilde{E}_2^c\cap \tilde{E}_3^c)+\sum_{j\in M_k} \P(E_{3j}|\tilde{E}_4^c) \label{eqn:error1_events}
\end{align}
where the events are defined as
\begin{align*}
E_1&=\{(U_k^n(1),X_k^n(1|1),\hat{Y}_{n_k}^n(\tilde{V}_{n_k}|1))\notin T_{\epsilon''}\}\cr
E_2&=\left\{(1,\tilde{V}_{n_k})\notin F_{n_k,d}\right\}\cr
E_{3j}&=\left\{(1,1)\notin F_{j,d}\right\} \mbox{ for } j\in M_k\cr
\tilde{E}_1&=\{(U_k^n(1),Y_{n_k}^n, \hat{Y}_{n_k}^n({v_{n_k}}|1)) \notin T_{\epsilon'} \mbox{ for all } {v_{n_k}}\in [1:2^{nr_{n_k,v}}] \}\cr
\tilde{E}_2&=\left\{(U_k^n(1),Y_{n_k}^n)\notin T_{\epsilon'} \right\}\cr
\tilde{E}_3&=\left\{(U_k^n(\alpha_k),Y_{n_k}^n)\in T_{\epsilon'} \mbox{ for some } \alpha_k\neq 1\right\}\cr
\tilde{E}_4&=\left\{X_k^n(\beta_k|\alpha_k)=X_k^n(1|1) \mbox{ for some } (\alpha_k,\beta_k)\neq (1,1)\right\}.
\end{align*}
Note that $\tilde{E}_1^c$ implies that $(U_k^n(1),Y_{n_k}^n, \hat{Y}_{n_k}^n({\tilde{V}_{n_k}}|1)) \in T_{\epsilon'}$, $\tilde{E}_2^c\cap \tilde{E}_3^c$ implies that $\tilde{W}_{n_k}=(1,\tilde{V}_{n_k})$, and $\tilde{E}_4^c$ implies that $\tilde{W}_{j}=(1,1)$ for all $j\in M_k$. Let us upper bound each term in the right-hand side of (\ref{eqn:error1_events}).
\begin{itemize}
\item If $r_{n_k,v}>I(Y_{n_k};\hat{Y}_{n_k}|U_k)+\delta(\epsilon')$,\footnote{Here and from now on, $\delta(\epsilon')\rightarrow 0$ as $\epsilon'\rightarrow 0$.} we have $\P(\tilde{E}_1)<\epsilon$ for sufficiently large $n$ from the covering lemma~\cite{GamalKim:09}.
\item By the law of large numbers,  we have $\P(\tilde{E}_2)<\epsilon$ for sufficiently large $n$.
\item If $r_{k,a}<I(U_k;Y_{n_k})-\delta(\epsilon')$, we have $\P(\tilde{E}_3)<\epsilon$ for sufficiently large $n$ from the packing lemma~\cite{GamalKim:09}.
\item If $r_{k,a}+r_{k,b}<H(X_k)-\delta(\epsilon')$ and $r_{k,b}<H(X_k|U_k)-\delta(\epsilon')$, we have $\P(\tilde{E}_4)<\epsilon$ for sufficiently large $n$.
\item We have
\begin{align*}
&\P(E_1\cap\tilde{E}_1^c)\\
&=\P\{(U_k^n(1),X_k^n(1|1),\hat{Y}_{n_k}^n(\tilde{V}_{n_k}|1))\notin T_{\epsilon''}, (U_k^n(1),Y_{n_k}^n, \hat{Y}_{n_k}^n(\tilde{V}_{n_k}|1)) \in T_{\epsilon'}\}\\
&\leq \sum_{(u_k^n,y_{n_k}^n,\hat{y}_{n_k}^n)\in T_{\epsilon'}}p(u_k^n,y_{n_k}^n,\hat{y}_{n_k}^n)P\{(u_k^n(1),X_k^n(1|1),\hat{y}_{n_k}^n(\tilde{V}_{n_k}|1))\notin T_{\epsilon''}|u_k^n,y_{n_k}^n,\hat{y}_{n_k}^n\}\\
&\overset{(a)}{\leq}\epsilon
\end{align*}
for sufficiently large $n$, where $(a)$ is from the conditional typicality lemma~\cite{GamalKim:09}.
\item We have $\P(E_2|\tilde{E}_2^c\cap \tilde{E}_3^c)=\mu_{n_k,d}^{(n)}<\epsilon$
for sufficiently large $n$.
\item We have $\sum_{j\in M_k} \P(E_{3j}|\tilde{E}_4^c)=\sum_{j\in M_k}\mu_{j,d}^{(n)}<\epsilon$ for sufficiently large $n$.
\end{itemize}
Let us choose $r_{k,a}$, $r_{k,b}$ and $r_{n_k,v}$ as
\begin{align*}
r_{k,a}&=I(U_k;Y_{n_k})-2\delta(\epsilon')\\
r_{k,b}&=H(X_k|U_k)-2\delta(\epsilon')\\
r_{n_k,v}&=I(Y_{n_k};\hat{Y}_{n_k}|U_k)+2\delta(\epsilon').
\end{align*}
For the above choice of $r_{k,a}, r_{k,b}$, and $r_{n_k,v}$, we have $\mu_{k,d}^{(n)}<7\epsilon$ for sufficiently large $n$.

Now, consider $\tilde{w}_k'\neq (1,1)$. $\nu_{k,d}^{(n)}$ is upper-bounded as
\begin{align}
\nu_{k,d}^{(n)}&=\P(\tilde{w}_k'\in F_{k,d})\cr
&\leq \P(E_4 \cup E_5 \cup E_6)\cr
&\leq \P(\tilde{E}_2 \cup \tilde{E}_3 \cup \tilde{E}_4\cup E_4 \cup E_5 \cup E_6 )\cr
&\leq \P(\tilde{E}_2)+\P(\tilde{E}_3)+\P(\tilde{E}_4)+\P(E_4)+\P(E_5\cap \tilde{E}_2^c\cap \tilde{E}_3^c\cap \tilde{E}_4^c)+\P(E_6\cap\tilde{E}_2^c\cap \tilde{E}_3^c\cap \tilde{E}_4^c) \cr
&< 3\epsilon +\P(E_4)+\P(E_5\cap \tilde{E}_2^c\cap \tilde{E}_3^c\cap \tilde{E}_4^c)+\P(E_6\cap\tilde{E}_2^c\cap \tilde{E}_3^c\cap \tilde{E}_4^c)\label{eqn:error2_events}
\end{align}
for sufficiently large $n$, where the events are given as
\begin{align*}
E_4&=\left\{\gamma_k(\tilde{w}_k')=(1,1)\right\}\cr
E_5&=\{\gamma_k(\tilde{w}_k')=(1, \beta_k),(U_k^n(1),X_k^n(\beta_k|1),\hat{Y}_{n_k}^n(\tilde{V}_{n_k}|1))\in T_{\epsilon''}, (1,\tilde{V}_{n_k})\in F_{n_k,d}, \cr
&~~~~~~~~~~~~~~~~~~~~~~~~~~~~~~~~~~~\left.(1,\beta_k)\in F_{j,d} \mbox{ for all } j\in M_k \mbox{ for some }\beta_k\neq 1\right\}\cr
E_6&= \{\gamma_k(\tilde{w}_k')=(\alpha_k, \beta_k),(U_k^n(\alpha_k),X_k^n(\beta_k|\alpha_k),\hat{Y}_{n_k}^n({v_{n_k}}|\alpha_k))\in T_{\epsilon''}, (\alpha_k,{v_{n_k}})\in F_{n_k,d},\cr
&~~~~~~(\alpha_k,\beta_k)\in F_{j,d} \mbox{ for all } j\in M_k \mbox{ for some }(\alpha_k, \beta_k)\neq (1,1)\mbox{ and }(\alpha_k, {v_{n_k}})\neq (1,\tilde{V}_{n_k})\}.
\end{align*}
Let us upper bound each term in the right-hand side of (\ref{eqn:error2_events}).
\begin{itemize}
\item $\P(E_4)$ is given as
\begin{align*}
\P(E_4)=2^{-nr_{k,a}}2^{-nr_{k,b}}=2^{-n(I(U_k;Y_{n_k})+H(X_k|U_k)-4\delta(\epsilon'))}.
\end{align*}
\item We have
\begin{align*}
&\P(E_5\cap \tilde{E}_2^c\cap \tilde{E}_3^c\cap \tilde{E}_4^c)\cr
&\leq\sum_{\beta_k\neq 1}\P(\gamma_k(\tilde{w}_k')=(1, \beta_k))\P((U_k^n(1),X_k^n(\beta_k|1),\hat{Y}_{n_k}^n(\tilde{V}_{n_k}|1))\in T_{\epsilon''})\prod_{j\in M_k}\nu_{j,d}^{(n)}\\
&\overset{(a)}{<} 2^{nr_{k,b}}2^{-n(r_{k,a}+r_{k,b})}2^{-n(I(X_k;\hat{Y}_{n_k}|U_k)-\delta(\epsilon''))}2^{-n(\sum_{j\in M_k}r_j-\epsilon)}\\
&= 2^{-n(\sum_{j\in M_k}r_j+I(U_k;Y_{n_k})+I(X_k;\hat{Y}_{n_k}|U_k)-2\delta(\epsilon')-\delta(\epsilon'')-\epsilon)}\nonumber
\end{align*}
for sufficiently large $n$, where $(a)$ is because
\begin{align*}
\P((U_k^n(1),X_k^n(\beta_k|1),\hat{Y}_{n_k}^n(\tilde{V}_{n_k}|1))\in T_{\epsilon''})<2^{-n(I(X_k;\hat{Y}_{n_k}|U_k)-\delta(\epsilon''))}
\end{align*}
for $\beta_k\neq 1$ from the joint typicality lemma~\cite{GamalKim:09}.
\item We get
\begin{align*}
&\P(E_6\cap\tilde{E}_2^c\cap \tilde{E}_3^c\cap \tilde{E}_4^c)\cr
&\leq\sum_{\substack{\alpha_k,\beta_k,{v_{n_k}}\\ (\alpha_k,\beta_k)\neq (1,1)\\ (\alpha_k,{v_{n_k}})\neq (1,\tilde{V}_{n_k})}} \P(\gamma_k(\tilde{w}_k')=(\alpha_k, \beta_k))\P((U_k^n(\alpha_k),X_k^n(\beta_k|\alpha_k),\hat{Y}_{n_k}^n({v_{n_k}}|\alpha_k))\in T_{\epsilon''})
\prod_{j\in Z_k}\nu_{j,d}^{(n)}\cr
&\overset{(a)}{<}2^{n(r_{k,a}+r_{k,b}+r_{n_k,v})}2^{-n(r_{k,a}+r_{k,b})}2^{-n(I(X_k;\hat{Y}_{n_k}|U_k)-\delta(\epsilon''))}2^{-n(\sum_{j\in Z_k}r_j-\epsilon)}\\
&= 2^{-n(\sum_{j\in Z_k}r_j-I(Y_{n_k};\hat{Y}_{n_k}|U_k,X_k)-2\delta(\epsilon')-\delta(\epsilon'')-\epsilon)}\nonumber
\end{align*}
for sufficiently large $n$, where $(a)$ is from the joint typicality lemma~\cite{GamalKim:09}.
\end{itemize}
Note that $r_j=0$ for $j\notin G_d$. Thus, we have
\begin{align*}
&\nu_{k,d}^{(n)}\cr
&<2^{-n(\min\{I(U_k;Y_{n_k})+H(X_k|U_k), \sum_{j\in M_k\cap G_d}r_j+I(U_k;Y_{n_k})+I(X_k;\hat{Y}_{n_k}|U_k), \sum_{j\in Z_k\cap G_d}r_j-I(Y_{n_k};\hat{Y}_{n_k}|U_k, X_k)
\}-2\epsilon)}
\end{align*}
for sufficiently small $\epsilon'$ and $\epsilon''$ and sufficiently large $n$.

\section{Upper Bound} \label{sec:ub}
Fix $d\in [1:K]$. Let $U_{k,i}\triangleq \left(X_{k,i+1}^n, Y_{n_k}^{i-1}\right)$ and $\hat{Y}_{n_k,i}\triangleq Y_{L_{n_k}\cap D_d}^n$ for $k\in [1:N]$ and $i\in [1:n]$. Note that
\begin{align*}
p\left(u_{k,i},x_{k,i},y_{n_k,i},\hat{y}_{n_k,i}\right)=p\left(u_{k,i},x_{k,i}\right)p\left(y_{n_k,i}|x_{k,i}\right)p\left(\hat{y}_{n_k,i}|u_{k,i},y_{n_k,i}\right)
\end{align*}
for $k\in [1:N]$ and $i\in [1:n]$. Consider a cut $S_d$ considered in Theorem \ref{thm:tree_multicast}.

Let us first present two lemmas and a corollary.
\begin{lemma} \label{lemma:basic}
For $k\in [1:N]$, the following inequalities and equality hold.
\begin{subequations}\label{eqn:cov_basic}
\begin{align}
\sum_{i=1}^n I(U_{k,i};Y_{n_k,i}) + H(X_{k,i}|U_{k,i}) - H(X_k^n)&\geq 0\label{eqn:cov_basic1}\\
\sum_{i=1}^n I(U_{k,i};Y_{n_k,i}) + I(X_{k,i};\hat{Y}_{n_k,i}|U_{k,i}) - I(X_k^n;Y_{L_{n_k}\cap D_d}^n)&\geq 0\label{eqn:cov_basic2}\\
-\sum_{i=1}^n I(Y_{n_k,i};\hat{Y}_{n_k,i}|U_{k,i}, X_{k,i}) + I(Y_{n_k}^n;Y_{L_{n_k}\cap D_d}^n|X_k^n)&=0\label{eqn:cov_basic3}
\end{align}
\end{subequations}
\end{lemma}
\begin{lemma} \label{lemma:coverse_recursive}
The following inequalities hold.
\begin{subequations}\label{eqn:cov_rec}
\begin{align}
I(X_k^n;Y_{L_k\cap D_d}^n)-H(X_k^n)&\leq 0 && \mbox{for $k\in [1:N]$} \label{eqn:cov_rec_1}\\
I(X_k^n;Y_{L_k\cap D_d}^n)-I(X_k^n;Y_{L_{n_k}\cap D_d}^n)&\leq \sum_{j\in M_k\cap G_d} I(X_j^n;Y_{L_j\cap D_d}^n) && \mbox{for $k\in B_{S_d,d}$} \label{eqn:cov_rec_2}\\
I(X_k^n;Y_{L_k\cap D_d}^n)+I(Y_{n_k}^n;Y_{L_{n_k}\cap D_d}^n|X_k^n)&\leq \sum_{j\in Z_k\cap G_d} I(X_j^n;Y_{L_j\cap D_d}^n)  && \mbox{for $k\in C_{S_d,d}$}  \label{eqn:cov_rec_3}
\end{align}
\end{subequations}
\end{lemma}
The proofs of Lemmas \ref{lemma:basic} and \ref{lemma:coverse_recursive} are in Appendices~\ref{appendix:basic} and~\ref{appendix:converse_recursive}, respectively. From Lemmas \ref{lemma:basic} and \ref{lemma:coverse_recursive}, we have the following corollary.

\begin{corollary}\label{corollary:converse}
We have
\begin{align*}
I(X_1^n;Y^n_{D_d})&\leq \sum_{k\in A_{S_d,d}} \sum_{i=1}^n I(U_{k,i};Y_{n_k,i})  +  H(X_{k,i}|U_{k,i})\cr
&+\sum_{k\in B_{S_d,d}} \sum_{i=1}^n I(U_{k,i};Y_{n_k,i}) +  I(X_{k,i};\hat{Y}_{n_k,i}|U_{k,i}) -  \sum_{k\in C_{S_d,d}} \sum_{i=1}^n I(Y_{n_k,i};\hat{Y}_{n_k,i}|U_{k,i},X_{k,i}).
\end{align*}
\end{corollary}
\begin{proof}
We have
\begin{align*}
I(X_1^n;Y^n_{D_d})\leq I(X_1^n;Y^n_{D_d})+\sum_{k\in S_d}\psi(k)+ \sum_{k\in A_{S_d,d}} \sum_{i=1}^n I(U_{k,i};Y_{n_k,i})  +  H(X_{k,i}|U_{k,i})\cr
+\sum_{k\in B_{S_d,d}} \sum_{i=1}^n I(U_{k,i};Y_{n_k,i}) +  I(X_{k,i};\hat{Y}_{n_k,i}|U_{k,i})-  \sum_{k\in C_{S_d,d}} \sum_{i=1}^n I(Y_{n_k,i};\hat{Y}_{n_k,i}|U_{k,i},X_{k,i})
\end{align*}
from Lemma \ref{lemma:basic}, where $\psi(k)$ for $k\in S_d$ is defined as
\begin{align*}
\psi(k)\triangleq
\begin{cases}
-H(X_k^n) &\mbox{if $k\in A_{S_d,d}$}\\
-I(X_k^n;Y_{L_{n_k}\cap D_d}^n) &\mbox{if $k\in B_{S_d,d}$}\\
I(Y_{n_k}^n;Y_{L_{n_k}\cap D_d}^n|X_k^n) &\mbox{if $k\in C_{S_d,d}$}
\end{cases}.
\end{align*}
Now, it remains to show
\begin{align}
I(X_1^n;Y^n_{D_d})+\sum_{k\in S_d}\psi(k)\leq 0. \label{eqn:coro_m}
\end{align}
From Lemma \ref{lemma:coverse_recursive}, we have
\begin{align}
I(X_k^n;Y_{L_k\cap D_d}^n)+\psi(k)\leq \sum_{j\in Z_k\cap S_d} I(X_j^n;Y_{L_j\cap D_d}^n) \label{eqn:bounding_re}
\end{align}
for $k\in S_d$. Using the inequality (\ref{eqn:bounding_re}) recursively for all $k\in S_d$ starting from $k=1$, the inequality (\ref{eqn:coro_m}) is proved from the fact that node $k$ at the boundary of $S_d$ is included in $A_{S_d,d}$ and $Z_k\cap S_d=\emptyset$ for $k\in A_{S_d,d}$.
\end{proof}

Now, we are ready to prove the upper bound in Theorem \ref{thm:tree_multicast}. In the following,  $\epsilon_n$ tends to zero as $n$ tends to infinity. We have
\begin{align}
nR&= H(X_1^n)\cr
&=I(X_1^n;Y^n_{D_d})+H(X_1^n|Y^n_{D_d})\cr
&\overset{(a)}{\leq} I(X_1^n;Y^n_{D_d})+n\epsilon_n \cr
&\overset{(b)}{\leq}n\epsilon_n+ \sum_{k\in A_{S_d,d}} \sum_{i=1}^n I(U_{k,i};Y_{n_k,i})  +  H(X_{k,i}|U_{k,i})\cr &+\sum_{k\in B_{S_d,d}} \sum_{i=1}^n I(U_{k,i};Y_{n_k,i}) +  I(X_{k,i};\hat{Y}_{n_k,i}|U_{k,i})  \cr
&-  \sum_{k\in C_{S_d,d}} \sum_{i=1}^n I(Y_{n_k,i};\hat{Y}_{n_k,i}|U_{k,i},X_{k,i})\nonumber
\end{align}
where $(a)$ is due to Fano's inequaility and $(b)$ is from Corollary \ref{corollary:converse}.

Let $Q$ denote a time-sharing random variable uniformly distributed over $[1:n]$ that is independent of all the other variables. Define random variables $(U_k', X_k,Y_{n_k},\hat{Y}_{n_k}')$ for $k\in [1:N]$ such that
\begin{align*}
&p\left(U_k'=u_k, X_k=x_k,Y_{n_k}=y_{n_k},\hat{Y}'_{n_k}=\hat{y}_{n_k}|Q=i\right)\cr
&=p\left(U_{k,i}=u_k, X_{k,i}=x_k,Y_{n_k,i}=y_{n_k},\hat{Y}_{n_k,i}=\hat{y}_{n_k}\right)
\end{align*}
for $i\in [1:n]$.
Let $U_k\triangleq (U_k', Q)$ and $\hat{Y}_{n_k}\triangleq(\hat{Y}_{n_k}', Q)$ for $k\in [1:N]$.  Then, we have
\begin{align*}
\frac{1}{n}\sum_{i=1}^n I(U_{k,i};Y_{n_k,i})  +  H(X_{k,i}|U_{k,i})&=I(U_k';Y_{n_k}|Q)+H(X_{k}|U'_{k},Q)\cr
&\leq I(U_k;Y_{n_k})+H(X_{k}|U_{k}),
\end{align*}
\begin{align*}
\frac{1}{n}\sum_{i=1}^n I(U_{k,i};Y_{n_k,i}) +  I(X_{k,i};\hat{Y}_{n_k,i}|U_{k,i})&=I(U_{k}';Y_{n_k}|Q) +  I(X_{k};\hat{Y}_{n_k}'|U_{k}',Q)\\
&\leq I(U_{k};Y_{n_k}) +  I(X_{k};\hat{Y}_{n_k}|U_{k}),
\end{align*}
and
\begin{align*}
\frac{1}{n}\sum_{i=1}^n I(Y_{n_k,i};\hat{Y}_{n_k,i}|U_{k,i},X_{k,i})&= I(Y_{n_k};\hat{Y}'_{n_k}|U_{k}',X_{k},Q)\\
&= I(Y_{n_k};\hat{Y}_{n_k}|U_{k},X_{k}).
\end{align*}
Hence, we get
\begin{align}
R-\epsilon_n \leq   \sum_{k\in A_{S_d,d}}I(U_{k};Y_{n_k})+H(X_{k}|U_{k})\cr
+  \sum_{k\in B_{S_d,d}}I(U_{k};Y_{n_k})+I(X_{k};\hat{Y}_{n_k}|U_{k})\cr
-\sum_{k\in C_{S_d,d}}I(Y_{n_k};\hat{Y}_{n_k}|U_{k},X_{k}).\label{eqn:cov_primitive}
\end{align}

Note that only the marginal distributions $p\left(u_{k},x_{k},y_{n_k},\hat{y}_{n_k}\right)$'s for $k\in [1:N]$ are needed to evaluate the right-hand side of (\ref{eqn:cov_primitive}). Thus, we do not lose generality  when we only consider the joint distribution of (\ref{eqn:joint_distr_general}). Since the definition of $\hat{Y}_{n_k}$ for $k\in [1:N]$ depends on $D_d$'s for $d\in [1:K]$,  the minimization over $d\in[1:K]$ has to be outside the maximization over $\prod_{k\in [1:N]}p(\hat{y}_{n_k}|u_k,y_{n_k})$, which results in the upper bound (\ref{eqn:multicast_cov}).
The cardinality bound (\ref{eqn:cardinality}) for $\mathcal{U}_{k}$ and $\mathcal{\hat{Y}}_{n_k}$ for $k\in [1:N]$ can be obtained in a similar way as in~\cite{Kaspi:82}.

\section{Diamond Networks} \label{sec:diamond}
In this section, we present an alternative capacity expression for a simple tree network with a single destination, called a diamond network, in which the root node has one noisy child node and one noiseless child node, each node at the second level has a single noiseless child node, and nodes at the third level form the destination. In the following, nodes 1, 2, and 3 are the source, noisy relay, and noiseless relay, respectively.

The capacity of diamond networks was first characterized by Kang and Ulukus~\cite{KangUlukus:11}.
\begin{theorem}[Kang and Ulukus~\cite{KangUlukus:11}]
The capacity of diamond networks is given as
\begin{align}
\max_{\substack{p(u_1,x_1)p(\hat{y}_2|y_2,u_1):\\ r_2\geq I(Y_2;\hat{Y}_2|U_1,X_1)\\ r_3\geq H(X_1|U_1,\hat{Y}_2)} } \min\{I(U_1;Y_2)+H(X_1|U_1),r_2+r_3-I(Y_2;\hat{Y}_2|U_1,X_1)\}\label{eqn:cap_org}
\end{align}
with cardinalities of alphabets bounded by
\begin{subequations} \label{eqn:cardinality_diamond}
\begin{align}
|\mathcal{U}_1|&\leq |\mathcal{X}_1|+4 \\
|\hat{\mathcal{Y}}_2|&\leq |\mathcal{U}_1||\mathcal{Y}_2|+2\leq |\mathcal{X}_1||\mathcal{Y}_2|+4|\mathcal{Y}_2|+2.
\end{align}
\end{subequations}
\end{theorem}

Now, the following theorem shows an alternative capacity expression for diamond networks, whose proof is in Appendix~\ref{appendix:alter}.
\begin{theorem}[Alternative expression]\label{thm:alter}
The capacity of diamond networks is given as
\begin{align}
\max_{\substack{p(u_1,x_1)p(\hat{y}_2|y_2,u_1):\\ r_3\geq H(X_1|U_1,\hat{Y}_2)\\r_2+r_3\geq I(U_1;Y_2)+H(X_1|U_1)+I(Y_2;\hat{Y}_2|U_1,X_1)} } I(U_1;Y_2)+H(X_1|U_1)\label{eqn:cap_alter}
\end{align}
with cardinalities of alphabets bounded by (\ref{eqn:cardinality_diamond}).
\end{theorem}
Theorem \ref{thm:alter} shows that we do not lose optimality when the codebook construction of the combination of DF and CF is restricted to  the superposition of $2^{n(I(U_1;Y_2)-\epsilon)}$ `cloud centers' $U_1^n$, i.e., the part of the message decoded by the noisy relay, and $2^{n(H(X_1|U_1)-\epsilon)}$ `satellites' $X_1^n$ for each $U_1^n$, i.e., the remaining part of the message. This means that the optimality of the combination of DF and CF at the noisy relay in diamond networks intuitively makes sense since the relay compresses a noisy observation of almost uncoded information that has no structure. Otherwise, the optimality of \emph{compression} after decoding at the noisy relay, which ignores the codebook structure at the source, would have been counterintuitive.

On the other hand, Theorem~\ref{thm:tree_multicast} gives the following min-cut capacity expression for diamond networks with cardinalities of alphabets bounded by (\ref{eqn:cardinality_diamond}).
\begin{align}
\max_{p(u_1,x_1)p(\hat{y}_2|y_2,u_1)} \min\{I(U_1;Y_2)+H(X_1|U_1),& r_3+I(U_1;Y_2)+I(X_1;\hat{Y}_2|U_1),\cr
&~~~~~~~~~~~r_2+r_3-I(Y_2;\hat{Y}_2|U_1,X_1)\} \label{eqn:cap_mincut}
\end{align}
We note that the relationship between the two capacity characterizations (\ref{eqn:cap_alter}) and (\ref{eqn:cap_mincut}) is similar to that between the two equivalent achievable rate characterizations of CF for 3-node relay networks in~\cite{Gamal:79}~and~\cite{GamalMohseniZahedi:06}, which are given by (\ref{eqn:general1}) and (\ref{eqn:general2}), respectively. Here, node indices follow the convention that nodes 1, 2, and 3 are the source, relay, and destination, respectively.
\begin{align}
\max_{\substack{p(x_1)p(x_2)p(\hat{y}_2|y_2,x_2): \\ I(X_2;Y_3)\geq I(Y_2;\hat{Y}_2|X_2,Y_3)}}&I(X_1;\hat{Y}_2,Y_3|X_2) \label{eqn:general1}\\
\max_{p(x_1)p(x_2)p(\hat{y}_2|y_2,x_2)} \min\{I(X_1;\hat{Y}_2,Y_3|X_2),&I(X_1,X_2;Y_3)-I(Y_2;\hat{Y}_2|X_1,X_2,Y_3)\}\label{eqn:general2}
\end{align}

\section{Conclusion} \label{sec:conclusion}
We characterized the capacity of a class of multicast tree networks having an arbitrary number of nodes, which includes the class of diamond networks studied in \cite{KangUlukus:11} as a special case. For achievability, we constructed a robust coding scheme that uses a combination of DF and CF in every noisy relay and a random binning in every noiseless relay in a way that the codebook constructions and relay operations are independent for each node. For converse, we used a novel technique of iteratively manipulating inequalities exploiting the tree topology. For diamond networks, we showed that the optimality of the combination of DF and CF at the noisy relay is intuitively convincing by proving that it does not lose optimality to restrict the coding scheme such that what is compressed after decoding at the noisy relay is a noisy observation of almost uncoded information.


\appendices

\section{Proof of Corollary \ref{corollary:tightness}} \label{appendix:tightness}
Let $C_{p_kk}\triangleq \max_{p(x_{p_k})}I(X_{p_k};Y_k)$ for $k\in [2:N]$ denote the point-to-point capacity between nodes $p_k$ and $k$ and let $C(k,d)$ for $d\in[1:K]$ and $k$ such that $L_k\supseteq D_d$ denote the capacity of tree network $T_{k}$ with a source $k$ and a single destination $D_d$. For a lower bound on the right-hand side of (\ref{eqn:multicast_ach}), let us choose the joint distribution $\prod_{k\in [1:N]}p(u_k,x_k)p(\hat{y}_{n_k}|u_k,y_{n_k})$ as follows:
\begin{itemize}
\item For $k$ such that $k\in T_{a_d}$ for some $d\in[1:K]$, choose $p(u_k,x_k)p(\hat{y}_{n_k}|u_k,y_{n_k})$ that achieves $C(a_d, d)$.
\item For $k$ such that $k\notin T_{a_d}$ for all $d\in [1:K]$ and $n_k\neq \emptyset$, choose $p(x_k)$ that achieves $C_{kn_k}$ and let $U_k=X_k$ and $\hat{Y}_{n_k}=\emptyset$.
\item For $k$ such that $k\notin T_{a_d}$ for all $d\in [1:K]$ and $n_k= \emptyset$, let $X_k$ uniformly distributed over $\mathcal{X}_k$ and let $U_k=\hat{Y}_{n_k}=\emptyset$.
\end{itemize}
For the above choice of distribution, we obtain the following lower bound.
\begin{align*}
C\geq&\min \left\{\min_{k\notin \bigcup_{d\in [1:K]}T_{a_d}}\min_{j\in Z_k}C_{kj}, \min_{d\in [1:K]} C(a_d,d)\right\}\cr
=&\min \left\{\min_{d\in [1:K]}\min_{k\in G_d \cap T_{a_d}^c}\min_{j\in Z_k\cap G_d}C_{kj}, \min_{d\in [1:K]} C(a_d,d)\right\}\cr
=&\min_{d\in [1:K]}\min  \left\{\min_{k\in G_d \cap T_{a_d}^c}\min_{j\in Z_k\cap G_d}C_{kj}, C(a_d,d)\right\}\cr
=&\min_{d\in [1:K]}C(1, d).
\end{align*}

Now, note that the right-hand side of (\ref{eqn:multicast_cov}) is clearly upper-bounded by $\min_{d\in [1:K]}C(1, d)$. Hence, the lower and upper bounds in Theorem \ref{thm:tree_multicast} coincide. \endproof

\section{Proof of Lemma~\ref{lemma:basic}}\label{appendix:basic}
Consider $k\in [1:N]$. We have
\begin{align*}
&H(X_k^n)\cr
&=\sum_{i=1}^{n}H(X_{k,i}|X_{k,i+1}^n)\cr
&\leq \sum_{i=1}^{n}I(Y_{n_k}^{i-1};Y_{n_k,i})+H(X_{k,i}|X_{k,i+1}^n)\cr
&=\sum_{i=1}^{n}I(X_{k, i+1}^n,Y_{n_k}^{i-1};Y_{n_k,i})-I(X_{k, i+1}^{n};Y_{n_k,i}|Y_{n_k}^{i-1})+H(X_{k,i}|X_{k,i+1}^n,Y_{n_k}^{i-1})+I(X_{k,i};Y_{n_k}^{i-1}|X_{k,i+1}^n)\cr
&\overset{(a)}{=}\sum_{i=1}^{n}I(X_{k, i+1}^n,Y_{n_k}^{i-1};Y_{n_k,i})+H(X_{k,i}|X_{k,i+1}^n,Y_{n_k}^{i-1})\\
&=\sum_{i=1}^{n}I(U_{k,i};Y_{n_k,i})+H(X_{k,i}|U_{k,i})
\end{align*}
where $(a)$ is from Csisz\'{a}r sum identity~\cite{CsiszarKorner:78}, which proves (\ref{eqn:cov_basic1}).

We have
\begin{align}
H(X_k^n|Y^n_{L_{n_k}\cap D_d})&=\sum^n_{i=1}H(X_{k,i}|X_{k,i+1}^n, Y^n_{L_{n_k}\cap D_d})\cr
&\geq \sum^n_{i=1}H(X_{k,i}|X_{k,i+1}^n, Y_{n_k}^{i-1}, Y^n_{L_{n_k}\cap D_d})\cr
&= \sum^n_{i=1}H(X_{k,i}|U_{k,i}, \hat{Y}_{{n_k},i}).\label{eqn:basic_second_proof}
\end{align}
Note that combining (\ref{eqn:basic_second_proof}) with (\ref{eqn:cov_basic1}) proves (\ref{eqn:cov_basic2}).

We have
\begin{align*}
I(Y_{n_k}^n;Y_{L_{n_k}\cap D_d}^n|X_k^n)&=\sum_{i=1}^nI(Y_{n_k,i};Y_{L_{n_k}\cap D_d}^n|X_k^n,Y_{n_k}^{i-1})\cr
&\overset{(a)}{=}\sum_{i=1}^nI(Y_{n_k,i};Y_{L_{n_k}\cap D_d}^n|X_{k,i}^n,Y_{n_k}^{i-1})\\
&=\sum_{i=1}^nI(Y_{n_k,i};\hat{Y}_{n_k,i }|U_{k,i},X_{k,i})\nonumber
\end{align*}
where $(a)$ is from the following Markov chains:
\begin{align*}
X_k^{i-1}&\leftrightarrow(X_{k,i}^n,Y_{n_k}^{i-1})\leftrightarrow Y_{L_{n_k}\cap D_d}^n\\
X_k^{i-1}&\leftrightarrow(X_{k,i}^n,Y_{n_k}^{i})\leftrightarrow Y_{L_{n_k}\cap D_d}^n,
\end{align*}
which proves (\ref{eqn:cov_basic3}).
\endproof

\section{Proof of Lemma~\ref{lemma:coverse_recursive}}\label{appendix:converse_recursive}
For $k\in [1:N]$, the inequality (\ref{eqn:cov_rec_1}) holds trivially.

For $k\in B_{S_d,d}$, we have
\begin{align*}
I(X_k^n;Y_{L_k\cap D_d}^n)-I(X_k^n;Y_{L_{n_k}\cap D_d}^n)&=I(X_k^n;Y_{L_k\cap L_{n_k}^c\cap D_d}^n|Y_{L_{n_k}\cap D_d}^n)\cr
&\overset{(a)}{\leq} I(X_k^n;Y_{L_k\cap L_{n_k}^c\cap D_d}^n)\\
&=I(X_k^n;Y_{\bigcup_{j\in M_k}(L_j\cap D_d)}^n)\\
&\overset{(b)}{\leq} \sum_{j\in M_k}I(X_k^n;Y_{L_j\cap D_d}^n)\\
&\overset{(c)}{=} \sum_{j\in M_k\cap G_d}I(X_k^n;Y_{L_j\cap D_d}^n)\\
&\overset{(d)}{\leq} \sum_{j\in M_k\cap G_d}I(X_j^n;Y_{L_j\cap D_d}^n)
\end{align*}
where $(a)$ is from the Markov chain
\begin{align}
Y_{L_k\cap L_{n_k}^c\cap D_d}^n\leftrightarrow X_k^n \leftrightarrow Y_{L_{n_k}\cap D_d}^n, \label{eqn:cov_rec_111}
\end{align}
$(b)$ is from the Markov chain
\begin{align}
Y_{\bigcup_{m\in M_k, m<j}L_m\cap D_d}^n\leftrightarrow X_k^n\leftrightarrow Y_{L_j\cap D_d}^n \label{eqn:cov_rec_121}
\end{align}
for $j\in M_k$, $(c)$ is because $L_j\cap D_d=\emptyset$ for $j\notin G_d$, and $(d)$ is from the following Markov chain
\begin{align}
X_k^n\leftrightarrow X_j^n\leftrightarrow Y_{L_j\cap D_d}^n \label{eqn:cov_rec_131}
\end{align}
for $j\in M_k\cap G_d$.
Note that (\ref{eqn:cov_rec_131}) holds since $j\in M_k\cap G_d$ is not a leaf node from the definition of $B_{S_d,d}$. Thus, (\ref{eqn:cov_rec_2}) is proved.

For $k\in C_{S_d,d}$, we get
\begin{align*}
&I(X_k^n;Y_{L_k\cap D_d}^n)+I(Y_{n_k}^n;Y_{L_{n_k}\cap D_d}^n|X_k^n)\cr
&=I(X_k^n;Y_{L_{n_k}\cap D_d}^n)+I(X_k^n;Y_{L_k\cap L_{n_k}^c\cap D_d}^n|Y_{L_{n_k}\cap D_d}^n)+I(Y_{n_k}^n;Y_{L_{n_k}\cap D_d}^n|X_k^n)\cr
&=I(X_k^n,Y_{n_k}^n;Y_{L_{n_k}\cap D_d}^n)+I(X_k^n;Y_{L_k\cap L_{n_k}^c\cap D_d}^n|Y_{L_{n_k}\cap D_d}^n)\cr
&\overset{(a)}{=}I(Y_{n_k}^n;Y_{L_{n_k}\cap D_d}^n)+I(X_k^n;Y_{L_k\cap L_{n_k}^c\cap D_d}^n|Y_{L_{n_k}\cap D_d}^n)\\
&\overset{(b)}{\leq} I(Y_{n_k}^n;Y_{L_{n_k}\cap D_d}^n)+\sum_{j\in M_k}I(X_k^n;Y_{L_j\cap D_d}^n)\\
&\overset{(c)}{\leq} I(Y_{n_k\cap G_d}^n;Y_{L_{n_k}\cap D_d}^n)+\sum_{j\in M_k\cap G_d}I(X_k^n;Y_{L_j\cap D_d}^n)\\
&\overset{(d)}{\leq} \sum_{j\in Z_k\cap G_d}I(X_j^n;Y_{L_j\cap D_d}^n)
\end{align*}
where $(a)$ is from the Markov chain
\begin{align*}
X_k^n\leftrightarrow Y_{n_k}^n\leftrightarrow Y_{L_{n_k}\cap D_d}^n,
\end{align*}
$(b)$ is from the Markov chains (\ref{eqn:cov_rec_111}) and (\ref{eqn:cov_rec_121}), $(c)$ is because $L_j\cap D_d=\emptyset$ for $j\notin G_d$, and $(d)$ is from the Markov chains (\ref{eqn:cov_rec_131}) and
\begin{align}
Y_{n_k\cap G_d}^n\leftrightarrow X_{n_k\cap G_d}^n\leftrightarrow Y_{L_{n_k}\cap D_d}^n. \label{eqn:markov}
\end{align}
Note that (\ref{eqn:cov_rec_131}) and (\ref{eqn:markov}) hold since $n_k\cap G_d$ and $j\in M_k\cap G_d$ are not leaf nodes from the definition of $C_{S_d,d}$. Thus, (\ref{eqn:cov_rec_3}) is proved. \endproof

\section{Proof of Theorem~\ref{thm:alter}} \label{appendix:alter}
Let us note that the constraint on $r_2$ in (\ref{eqn:cap_org}) can be easily verified to be redundant. Fix $r_2$ and $r_3$. Let $R_1$ and $R_2$ denote (\ref{eqn:cap_org}) without the constraint on $r_2$ and (\ref{eqn:cap_alter}), respectively. It is trivial to show $R_2\leq R_1$. To show $R_1\leq R_2$, it is enough to show that for all $p(u_1,x_1)p(\hat{y_2}|u_1,y_2)$ such that
$R<I(U_1;Y_2)+H(X_1|U_1)$ and $r_3\geq H(X_1|U_1,\hat{Y}_2)$, where $R\triangleq r_2+r_3-I(Y_2;\hat{Y}_2|U_1,X_1)$, there exists $p(u_1^*,x_1^*)p(\hat{y_2}^*|u_1^*,y_2)$ that satisfies
\begin{subequations} \label{eqn:another}
\begin{align}
R&= I(U_1^*;Y_2)+H(X_1^*|U_1^*),\label{eqn:another1}\\
R&\leq r_2+r_3-I(Y_2;\hat{Y}_2^*|U_1^*, X_1^*), \label{eqn:another2}\\
r_3&\geq H(X_1^*|U_1^*,\hat{Y}_2^*).\label{eqn:another3}
\end{align}
\end{subequations}
Now, consider a joint distribution of $p(u_1,x_1)p(\hat{y_2}|u_1,y_2)$ such that $R<I(U_1;Y_2)+H(X_1|U_1)$ and $r_3\geq H(X_1|U_1,\hat{Y}_2)$. Let $B$ denote a Bernoulli random variable with parameter $\lambda\in [0,1]$. Let $(U_1'', X_1'', \hat{Y}_2'')$ and $(U_1''', X_1''', \hat{Y}_2''')$  denote the triplets of random variables given as
\begin{align*}
(U_1'', X_1'', \hat{Y}_2'')=
\begin{cases}
(U_1,X_1,\hat{Y}_2) &\mbox{ if $B=1$} \\
(X_1,X_1,\emptyset) &\mbox{ if $B=0$}
\end{cases}
,~(U_1''', X_1''', \hat{Y}_2''')=
\begin{cases}
(\emptyset,\emptyset,\emptyset)& \mbox{ if $B=1$} \\
(X_1,X_1,\emptyset) &\mbox{ if $B=0$}
\end{cases}.
\end{align*}
We will show the existence of  $p(u_1^*,x_1^*)p(\hat{y_2}^*|u_1^*,y_2)$ that satisfies (\ref{eqn:another}) separately for the cases of $R>I(X_1;Y_2)$ and $R\leq I(X_1;Y_2)$. First, consider the case of $R>I(X_1;Y_2)$. Let $U_1^*=(U_1'',B)$, $X_1^*=X_1''$, and $\hat{Y}_2^*=(\hat{Y}_2'',B)$. Note that $I(U_1^*; Y_2)+H(X_1^*|U_1^*)$ is a continuous function of $\lambda$ and becomes $I(U_1;Y_2)+H(X_1|U_1)$ and $I(X_1;Y_2)$ when $\lambda=1$ and $\lambda =0$, respectively. From the intermediate value theorem, there exists $\lambda\in [0,1]$ such that $R=I(U_1^*; Y_2)+H(X_1^*|U_1^*)$. Furthermore, (\ref{eqn:another2}) and (\ref{eqn:another3}) are satisfied from
\begin{align*}
I(Y_2;\hat{Y}_2^*| U_1^*,X_1^*)&=I(Y_2;\hat{Y}_2''|U_1'',X_1'', B)\cr
&=\lambda I(Y_2;\hat{Y}_2|U_1,X_1) \cr
&\leq I(Y_2;\hat{Y}_2|U_1,X_1)
\end{align*}
and
\begin{align*}
H(X_1^*|U_1^*,\hat{Y}_2^*)&=H(X_1''|U_1'',\hat{Y}_2'', B)\cr
&=\lambda H(X_1|U_1,\hat{Y}_2)  \cr
&\leq H(X_1|U_1,\hat{Y}_2),
\end{align*}
respectively.

Next, consider the case of $R\leq I(X_1;Y_2)$. Let $U_1^*=(U_1''',B)$, $X_1^*=X_1'''$, and $\hat{Y}_2^*=(\hat{Y}_2''',B)$. Note that $I(U_1^*; Y_2)+H(X_1^*|U_1^*)$ is a continuous function of $\lambda$ and becomes $0$ and $I(X_1;Y_2)$ when $\lambda=1$ and $\lambda =0$, respectively. From the intermediate value theorem, there exists $\lambda\in [0,1]$ such that $R= I(U_1^*; Y_2)+H(X_1^*|U_1^*)$. Furthermore, (\ref{eqn:another2}) and (\ref{eqn:another3}) are satisfied from
\begin{align*}
I(Y_2;\hat{Y}_2^*|U_1^*,X_1^*)=I(Y_2;\hat{Y}_2'''|U_1''', X_1''',B)=0
\end{align*}
and
\begin{align*}
H(X_1^*|U_1^*,\hat{Y}_2^*)=H(X_1'''|U_1''',\hat{Y}_2''', B)=0,
\end{align*}
respectively.
\endproof

\bibliographystyle{IEEEtran}

\begin{thebibliography}{10}
\providecommand{\url}[1]{#1}
\csname url@samestyle\endcsname
\providecommand{\newblock}{\relax}
\providecommand{\bibinfo}[2]{#2}
\providecommand{\BIBentrySTDinterwordspacing}{\spaceskip=0pt\relax}
\providecommand{\BIBentryALTinterwordstretchfactor}{4}
\providecommand{\BIBentryALTinterwordspacing}{\spaceskip=\fontdimen2\font plus
\BIBentryALTinterwordstretchfactor\fontdimen3\font minus
  \fontdimen4\font\relax}
\providecommand{\BIBforeignlanguage}[2]{{%
\expandafter\ifx\csname l@#1\endcsname\relax
\typeout{** WARNING: IEEEtran.bst: No hyphenation pattern has been}%
\typeout{** loaded for the language `#1'. Using the pattern for}%
\typeout{** the default language instead.}%
\else
\language=\csname l@#1\endcsname
\fi
#2}}
\providecommand{\BIBdecl}{\relax}
\BIBdecl

\bibitem{meulen_thesis}
E.~C. van~der Meulen, ``Transmission of information in a {T}-terminal discrete
  memoryless channel,'' Ph.D. dissertation, Univ. of California, Berkeley, CA,
  1968.

\bibitem{Meulen:71}
------, ``Three-terminal communication channels,'' \emph{Adv. Appl. Prob.},
  vol.~3, pp. 120--154, 1971.

\bibitem{Gamal:79}
T.~M. Cover and A.~{El Gamal}, ``Capacity theorems for the relay channel,''
  \emph{{IEEE} Trans. Inf. Theory}, vol.~25, pp. 572--584, Sep. 1979.

\bibitem{XieKumar:05}
L.-L. Xie and P.~R. Kumar, ``An achievable rate for the multiple-level relay
  channel,'' \emph{{IEEE} Trans. Inf. Theory}, vol.~51, pp. 1348--1358, Apr.
  2005.

\bibitem{KramerGastparGupta:05}
G.~Kramer, M.~Gastpar, and P.~Gupta, ``Cooperative strategies and capacity
  theorems for relay networks,'' \emph{{IEEE} Trans. Inf. Theory}, vol.~51, pp.
  3037--3063, Sep. 2005.

\bibitem{Kim:08}
Y.-H. Kim, ``Capacity of a class of deterministic relay channels,''
  \emph{{IEEE} Trans. Inf. Theory}, vol.~53, pp. 1328--1329, Mar. 2008.

\bibitem{Yu:09}
M.~Aleksic, P.~Razaghi, and W.~Yu, ``Capacity of a class of modulo-sum relay
  channels,'' \emph{{IEEE} Trans. Inf. Theory}, vol.~55, pp. 921--930, Mar.
  2009.

\bibitem{Lim:10}
S.~H. Lim, Y.-H. Kim, A.~{El Gamal}, and S.-Y. Chung, ``Noisy network coding,''
  \emph{{IEEE} Trans. Inf. Theory}, vol.~57, pp. 3132--3152, May 2011.

\bibitem{Yeung:00}
R.~Ahlswede, N.~Cai, S.~R. Li, and R.~W. Yeung, ``Network information flow,''
  \emph{{IEEE} Trans. Inf. Theory}, vol.~46, pp. 1204--1216, July 2000.

\bibitem{Effros:06}
A.~F. Dana, R.~Gowaikar, R.~Palanki, B.~Hassibi, and M.~Effros, ``Capacity of
  wireless erasure networks,'' \emph{{IEEE} Trans. Inf. Theory}, vol. 53(3),
  pp. 789--804, Mar. 2006.

\bibitem{AvestimehrDiggaviTse:11}
A.~S. Avestimehr, S.~N. Diggavi, and D.~Tse, ``Wireless network information
  flow: A deterministic approach,'' \emph{{IEEE} Trans. Inf. Theory}, vol.~57,
  pp. 1872--1905, April 2011.

\bibitem{KangUlukus:11}
W.~Kang and S.~Ulukus, ``Capacity of a class of diamond channels,''
  \emph{{IEEE} Trans. Inf. Theory}, vol.~57, pp. 4955--4960, Aug. 2011.

\bibitem{Orlitsky:01}
A.~Orlitsky and J.~R. Roche, ``Coding for computing,'' \emph{{IEEE} Trans. Inf.
  Theory}, vol.~47, pp. 903--917, Mar. 2001.

\bibitem{GamalKim:09}
A.~{El Gamal} and Y.-H. Kim, ``Lecture notes on network information theory,''
  [Online]. Available: http://arxiv.org/abs/1001.3404.

\bibitem{Kaspi:82}
A.~H. Kaspi and T.~Berger, ``Rate-distortion for correlated sources with
  partially separated encoders,'' \emph{{IEEE} Trans. Inf. Theory}, vol.~28,
  pp. 828--840, Nov. 1982.

\bibitem{GamalMohseniZahedi:06}
A.~{El Gamal}, M.~Mohseni, and S.~Zahedi, ``Bounds on capacity and minimum
  energy-per-bit for {AWGN} relay channels,'' \emph{{IEEE} Trans. Inf. Theory},
  vol.~52, pp. 1545--1561, Apr. 2006.

\bibitem{CsiszarKorner:78}
I.~Csisz{\'{a}}r and J.~K{\"{o}}rner, ``Broadcast channels with confidential
  messages,'' \emph{{IEEE} Trans. Inf. Theory}, vol.~24, pp. 339--348, Mar.
  1978.

\end{thebibliography}

\end{document}